\documentclass[11pt]{article}
\usepackage{amsfonts,amssymb,amsmath,latexsym,ae,aecompl}
\usepackage{graphics}
\usepackage{latexsym}
\usepackage{comment}
\usepackage{subcaption,multicol}
\usepackage[linesnumbered,ruled,vlined]{algorithm2e}

\usepackage{tikz}
\usepackage{caption}
\newcommand{\qed}{\hfill$\Box$}

\newenvironment{proof}{\noindent {\bf Proof.}}{\qed}

\newtheorem{theorem}{Theorem}[section]
\newtheorem{lemma}{Lemma}[section]
\newtheorem{corollary}{Corollary}[section]

\begin{document}

\baselineskip 0.2in
\parskip      0.1in
\parindent    0em

\bibliographystyle{plain}

\title{{\bf Collision-free Exploration by Mobile Agents Using Pebbles}}

\author{
Sajal K. Das \footnotemark[1]
\and
Amit Kumar Dhar \footnotemark[2]
\and
Barun Gorain \footnotemark[2]
\and
Madhuri Mahawar \footnotemark[2] 
}

\date{ }
\maketitle
\def\thefootnote{\fnsymbol{footnote}}

\footnotetext[1]{
\noindent
Department of Computer Science, Missouri University of Science and Technology, Rolla, USA.  E-mail:
{\tt sdas@mst.edu}}

\footnotetext[2]{
\noindent
Department of Computer Science, Indian Institute of Technology Bhilai, Bhilai, India. E-mails:
{\tt amitkdhar@iitbhilai.ac.in}, {\tt barun@iitbhilai.ac.in}, {\tt madhurim@iitbhilai.ac.in}
}

\begin{abstract}
 In this paper, we study collision-free graph exploration in an anonymous network. The network is modeled as a graph $G=(V, E)$ where the nodes of the graph are unlabeled, and each edge incident to a node $v$ has a unique label, called the port number, in $\{0,1,\cdots, d-1\}$, where $d$ is the degree of the node $v$. Two identical mobile agents, starting from different nodes in $G$ have to explore the nodes of $G$ in such a way that for every node $v$ in $G$, at least one mobile agent visits $v$ and no two agents are in the same node in any round and stop. The time of exploration is the minimum round number by which both agents have terminated. The agents know the size of the graph but do not know its topology. If an agent arrives in the one-hop neighborhood of the other agent, both agents can detect the presence of the other agent but have no idea at which neighboring node the other agent resides. The agents may wake up in different rounds, but once awake, execute a deterministic algorithm in synchronous rounds. An agent, after waking up, has no knowledge about the wake-up time of the other agent.

The task of collision-free exploration is impossible to solve even for a line of length 2 where the agents are placed at the end nodes of the line and even if both agents wake up at the same time. We study the problem of collision-free exploration where some pebbles are placed by an Oracle at the nodes of the graph to assist the agents in achieving collision-free exploration. The Oracle knows the graph, the starting positions of the agents, and their wake-up schedule, and it places some pebbles that may be of different colors, at most one at each node. The number of different colors of the pebbles placed by the Oracle is called the {\it color index} of the corresponding pebble placement algorithm.

The central question we study in this paper is as follows.

``What is the minimum number $z$ such that there exists a collision-free exploration of a given graph with pebble placement of color index $z$?''

For general graphs, we show that it is impossible to design a deterministic algorithm that achieves collision-free exploration with color index 1. We propose an exploration algorithm with color index 3. We also proposed a polynomial exploration algorithm for bipartite graphs with color index 2.



\end{abstract}

\pagebreak

\section{Introduction}
\subsection{Background}
Graph exploration by mobile agents is an extensively studied problem. In this problem, one or multiple mobile agents visit all the nodes or edges of a graph and have to declare that the exploration is completed. Exploring all the nodes of the graph is required when searching for data stored in the nodes while traversing all the edges is required when searching for a faulty edge in the network. In single-agent graph exploration, the exploration completes when all the nodes of the graph are visited by the agent. In multi-agent graph exploration, each node of the graph is visited by at least one mobile agent. Collision-free exploration problem by multiple mobile agents is studied in \cite{CzyzowiczDGKKP17, NakaminamiMH04}, where a set of mobile agents, starting from different nodes of the network, must explore the network without any collision: no two agents appear in the same node at the same time.

\subsection{Model and Problem Definition}
The network is modeled as a graph $G=(V, E)$, where the nodes in $V$ are unlabeled, where $|V| = n$. For each node $v \in V$, the edges incident to $v$ have unique labelings, called port labels in $\{0,1,\cdots,\delta-1\}$, where $\delta$ is the degree of the node $v$. Hence, each edge has two port labels at each of its incident nodes. An agent, initially located at a node $v$ only knows the degree of $v$. From any node $v$, if an agent takes the edge with the port number $p$ and reaches an adjacent node $w$, then it learns the port $q$, through which it entered $w$, and the degree of $w$. The timeline is divided into consecutive synchronous rounds, and at most one edge can be traversed by an agent in every round. Each round is divided into two different stages. In the first stage, an agent at a node $v$ is capable of doing any amount of local computations. In the second stage, the agent moves along one of the edge incidents to $v$ or decides to stay at $v$.

Two identical mobile agents, starting from different nodes in $G$ have to explore all nodes in $G$ in such a way that every node is visited by at least one mobile agent. Two agents are said to be in a \textit{collision} if the positions of the agents in any round are the same \cite{CzyzowiczDGKKP17}. However, if the two agents traverse the same edge in the same round from opposite directions, there is no collision \cite{CzyzowiczDGKKP17}. The objective is to design an exploration algorithm such that the movement paths of the agents are collision-free. The agents know the size of the graph, but do not know the topology of the graph. If an agent arrives in the one-hop neighborhood of the other agent, both agents can detect the presence of the other agent but have no idea at which neighboring node the other agent resides. Other than this `neighborhood sensing', the agents have no other means of communication between them. The agents may wake up in different rounds, but once awake, execute a deterministic algorithm in synchronous rounds. An agent, after waking up, does not know about the wake-up time of the other agent. Hence, the agents, executing a deterministic algorithm in synchronous rounds may have different local round numbers based on their different wake-up times (local round number starts from 1). An agent, after waking up, is not allowed to collide with the other agent even before its wake-up time.

The time of exploration is the minimum local round number by which both agents have terminated, and every node has been visited by at least one agent.

There are several real-life applications where collision-avoiding path planning is necessary. The recent surge of the Covid pandemic required individuals to provide essential services without physical interactions. In distributed systems where software agents updating data in different processes may require exclusive write access for consistency of data, collision-free movements of software agents over local area networks are necessary.

The task of collision-free exploration is not possible even for a line of length 2 where the agents are placed at the end nodes of the line even if the agents wake up at the same time. This is because as the graph is anonymous and the agents are identical, both agents either move or stay in some round $r$. When both move, they move to the same node, and a collision occurs. Otherwise, the middle node remains unexplored.

In order to assist the agents in achieving collision-free exploration, an Oracle places some pebbles at the nodes of the graph. An agent can see the status of a pebble placed at a node $v$ or not only when it visits $v$. Formally, let $\mathcal{C}$ be a set of colors and let $\mathcal{G}$ be the set of all port-labeled graphs. A pebble placement algorithm is a function $\mathcal{L}:\mathcal{I} \rightarrow \left\{\mathcal{C} \cup \{0\}\right\}^{|V|}$, where $\mathcal{I}$ is an instance of type $(G,v_1,v_2)$, $G \in \mathcal{G}$ and $v_1,v_2 \in G$. A node labeled by 0 using this function means that no pebble is placed on that node. Otherwise, the label $c_i$ means that a pebble of color $i$ is placed at the node. Note that either 0 or 1 pebble can be placed at every node of the graph. The size $|\mathcal{C}|$ of $\mathcal{C}$ is called the color index of $\mathcal{L}$. Similar Oracle models are defined in \cite{GorainMNP22} in the context of mobile agents and \cite{avery16} in the context of leader election in anonymous distributed networks.

In this paper, we ask what is the minimum color index that guarantees collision-free exploration. To be specific, the central question asked in this paper is as follows. 

``What is the minimum number $z$ such that there exists a collision-free exploration algorithm of a given graph with pebble placement of color index $z$?''
\vspace{1em}

\subsection{Novel Contributions} 

We show the existence of a collision-free exploration of any graph for pebble placement of color index 3.
In the case of a bipartite graph, our proposed algorithm works in polynomial time. We also show that it is impossible to design a collision free exploration algorithm for pebble placement with color index 1.


The collision-free exploration problem has been recently studied in \cite{bhagatpelc} using a slightly different model. Here, the graph is anonymous, but the agents are assumed to have two hop visibility. Like our current paper, the authors assumed the knowledge of $n$. The assumption of two hop visibility enables the authors to avoid collision, as whenever the agents come within distance two from one another, they can see the exact path that can lead to a collision. Also, this view enables them to break the symmetry between the identical agents, as the port numbers through which the agents can enter the node where a collision may happen are different for the agents. In our problem, we only assume 1 hop sensitivity, i.e., the agents only can sense one another if they come within one distance, but do not know at exactly which neighbor the other agent is present. Also, this does not eliminate the fact that the agents may arrive at distance two from one another and then may end up exploring the same node at the same time. We eliminated this possibility by carefully placing different colored pebbles in different nodes.


\subsection{Related Work}
The problem of exploration using mobile agents in unknown environments has been studied extensively where the environment is modeled in geometric pattern \cite{Bar-EliBFY94, BlumRS97, DengKP98} or as graphs where agents move along the edges. The graphs can be defined as strongly connected directed graph \cite{BenderFRSV02, BenderS94, DengP99} or as undirected graphs \cite{DereniowskiDKPU15, DiksFKP04, DuncanKK01, FraigniaudGKP06, PanaiteP99} or as labeled graphs \cite{CzyzowiczDGKKP17, DereniowskiDKPU15, PanaiteP99} where each node of the graphs has unique labels or anonymous\cite{BenderFRSV02, BenderS94}.
Graph exploration in an undirected graph $G$ of $n$ vertices and rooted at a node $r$ is studied in \cite{DereniowskiDKPU15}. A set of $k$ agents is initially located at $r$.  The exploration is completed when every vertex is visited by an agent. Two communication models are used for exploration: (i) \textit{global communication}, where at the end of each step $s$, all agents have complete knowledge of the explored subgraph, and (ii) \textit{local communication}, where two agents can exchange information only if they occupy the same vertex.
In \cite{NakaminamiMH04}, $n$ mobile agents visit $n$ nodes of a given network repeatedly, and no two agents occupy the same node at the same time. The authors studied self-stabilizing phase-based protocol for a tree network on a synchronous model. They obtained a bound of $\mathcal{O}(\Delta n)$ for agent traversal where $n$ is the number of nodes, and $\Delta$ is the maximum degree of any vertex. The analysis in \cite{NakaminamiMH04} is generalized in \cite{CzyzowiczDGKKP17} for arbitrary graphs, and also when agents do not have maps of the network. In \cite{CzyzowiczDGKKP17}, the multi-agent graph exploration problem assumes no two agents visit the same node at the same time, all the nodes have unique identifiers, and the agents are initially located at different nodes of the graph. The exploration problem is studied for trees and general graphs for two scenarios, (i) information about the network is given to all the agents, and (ii) no information about the network is provided to any agents. In the first scenario, a map of the network is provided as \textit{prior} knowledge to the agents. In \cite{CzyzowiczDKKKNO21}, graph exploration in an edge-weighted graph is studied. There are $k$ mobile agents placed at some vertices of the graph, and two agents can be placed at the same location. In this approach, every edge of the graph must be traversed by at least one agent. Each agent $r_i$ moves with the same speed and has a specific amount of energy equal to $e_i$ for its move. An agent can move only if $e_i>0$. Energy consumed by an agent is linearly proportional to the distance traveled. An $\mathcal{O}(n+k)$ algorithm finds a set of trajectories if exploration is possible. An $\mathcal{O}(n+lk^2)$ algorithm for $n-$ node tree is given to find an exploration strategy, where $l$ is the number of leaves in the tree. In \cite{BenderFRSV02}, exploration and mapping in an unknown environment using pebbles are studied. A robot is placed at a node in an anonymous strongly connected directed graph. The task of the robot is to explore and map the graph. It is shown that when the robot has prior knowledge about network size, then with one pebble, the graph can be explored in polynomial time, and when the robot has no prior knowledge about network size, then $\Theta(\log\log n)$ pebbles are necessary and sufficient to explore the graph.

Exploration in an anonymous undirected graph using pebbles is studied in \cite{DisserHK19}, where a single agent with constant memory is placed at a node of the graph. The agent has no prior knowledge about the network size. A tight bound of $\Theta(\log\log n)$ on the number of pebbles is shown to be necessary and sufficient to explore the graph. In \cite{Shah74}, the authors showed that with five pebbles, a finite automaton can search any arbitrary maze. In \cite{BlumS77}, authors showed that using four pebbles, a finite automaton can explore finite labyrinths. In \cite{BlumK78}, authors proved that with two pebbles, a finite automaton can search all the labyrinths. In \cite{Hoffmann81}, the author showed that with one pebble, a finite automaton cannot explore all finite labyrinths.

In \cite{GorainMNP22}, the treasure hunt problem in an anonymous graph by a mobile agent is solved in $\mathcal{O}(D\log \Delta+\log^3\Delta)$ time using $O (D \log \Delta)$ pebbles, where $D$ is the distance from the initial position of the agent to the treasure. 

In \cite{BhattacharyaGM22}, the trade-off between the number of pebbles and the time required to search a treasure in an anonymous graph by a mobile agent is studied for $k=D-1$ and $k=cD$, for $c$ positive integers where $k$ and $D$ are the number of pebbles and the distance of the treasure from the initial position of the agent, respectively. In \cite{saratdecentralized}, graph exploration and mapping using two robots in an unknown environment is studied.

\section{Collision-Free Exploration in General Graph}
\subsection{Impossibility result with color index 1}
In this section, we show an impossibility result of collision-free exploration with one colored pebble. To be specific, we prove that there does not exist any algorithm that guarantees collision-free exploration of two identical mobile agents. To prove this impossibility result, we have constructed a class of graphs $\mathcal{G}$, where irrespective of the pebble placement strategy and the exploration algorithm, either there is a collision or some node remains unexplored in at least one of the graphs in $\mathcal{G}$. We start by giving the details of the construction of the class of graph.

\textbf{Construction of the class of graphs $\mathcal{G}:$}

Let $G$ be a cycle of length four with vertices $x,x',y,y'$ in the clockwise order, and the port number 1 at $x'$ and $y'$ leads to the next node in the clockwise direction, and 0 leads to the anti-clockwise direction. The edges $(x,x')$ and $(y,x')$ have port number 0 at $x$ and $y$ and the edges $(x,y')$ and $(y,y')$ have port number 1 at $x$ and $y$.

We construct 3 graphs $G_0,G_1,G_2$ from $G$ as follows.
\begin{itemize}
    \item Add an edge between nodes $x$ and $y$ with port numbers 2 at each of its end nodes.
    \item For $i=0,1,2$, the graph $G_i$ is constructed by exchanging port number 2 with port number $i$ at both nodes $x$ and $y$.
\end{itemize}

Let $\mathcal{G}=\{ G_0, G_1, G_2\}$. (See Fig \ref{fig: lower}). We consider at each of the above input graphs, the initial positions of the two agents are the nodes $x$ and $y$, respectively.
\begin{figure}
    \centering
    \begin{subfigure}{0.3\textwidth}
        \centering
        \includegraphics{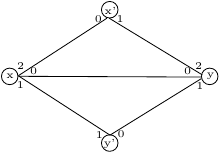}
        \caption{$G_0$}
    \end{subfigure}
    \hfill
    \begin{subfigure}{0.3\textwidth}
        \centering
        \includegraphics{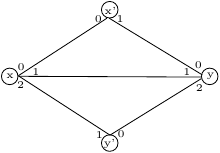}
        \caption{$G_1$}
    \end{subfigure}
    \hfill
    \begin{subfigure}{0.3\textwidth}
        \centering
        \includegraphics{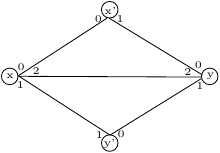}
        \caption{$G_2$}
    \end{subfigure}
    \caption{Graphs in class $\mathcal{G}$}
    \label{fig: lower}
\end{figure}

At any round, following any deterministic algorithm, an agent at the node $v$ can perform one of the following two steps: (a) stay at $v$ (b) move through an edge with the port number $ i \in \{0,1,2\}$. The desired impossibility result of this section is proved in two steps. In the first step, we show that if pebbles are placed at both $x$ and $y$ or pebbles are placed neither at $x$ nor $y$, then it is impossible to achieve collision-free explorations in the class of graphs $\mathcal{G}$. Then we show that if a pebble is placed either at $x$ or $y$, but not in both, then also collision-free exploration can not be achieved.

Let $\mathcal{A}$ be any exploration algorithm that solves collision-free exploration in the graphs in $\mathcal{G}$. Since the initial nodes of an agent have the same degrees in all the graphs in $\mathcal{G}$, after waking up, the steps performed by an agent following $\mathcal{A}$ depend only on whether a pebble is placed at its initial location or not. According to $\mathcal{A}$, based on whether a pebble is found at the initial position or not, the agent, after waking up, can take one of the following decisions:
\begin{enumerate}
    \item Stay at the current node and never move. We denote this decision by the agent as $Nevermove()$. 
    \item Stay at the current node until round $t \ge 1$ and move at round $t$ along port $i$ for $0\le i\le 2$. We denote this decision by the agent as $Move(i,t)$. 
\end{enumerate}

We show that for all possible decisions taken by the two mobile agents, there will be a graph where irrespective of the pebble placement by the Oracle, there will either be a collision or some node remains unexplored.

\begin{lemma}\label{lem:3_1}
    It is impossible to have a collision-free exploration algorithm for graphs in class $\mathcal{G}$ if pebbles are placed at both $x$ and $y$ or pebbles are placed neither at $x$ nor $y$.
\end{lemma}
    
\begin{proof}
    We show that in $G_0$ if pebbles are placed at both $x$ and $y$ or pebbles are placed neither at $x$ nor at $y$ then collision-free exploration is not possible to achieve. The proof for $G_1$ and $G_2$ are similar.
    
    Let $\mathcal{A}$ be any exploration algorithm. As both agents find a pebble at the initial node after waking up, or both agents see no pebble at the initial node after waking up, they make the same decisions. Based on the different possible decisions as stated above, we consider the following two cases.
    
    \begin{itemize}
        \item[Case 1] Both agents execute $Nevermove()$. In this case, both agents never leave their initial node and therefore the nodes $x'$ and $y'$ are never explored.
        \item[Case 2] Both agents execute $Move(i,t)$ for $i \in \{0,1,2\}$ and $t \ge 1$. When $i \ne 0$, if both agents wake up at the same time, then as per the construction of $G_0$ both agents taking port $i$ and move to the same node, {either $x'$ or $y'$}. This will result in a collision. When $i=0$, if the agent at $x$ wakes up at least one round before the agent at $y$, then at the local round number $t$ for the agent at $x$ moves to the node $y$ where the other agent has local round number $<t$ and therefore stays at $y$. Hence, a collision occurs at $y$. 
    \end{itemize}
\end{proof}

From lemma \ref{lem:3_1}, we can say that for every graph in class $\mathcal{G}$, at exactly one of the nodes $x$ and $y$, a pebble must be placed in order to avoid collision. The next lemma proves that even with this pebble placement, there exists at least one graph in $\mathcal{G}$, where collision-free exploration is not possible.

\begin{lemma}\label{lem:3_2}
    It is impossible to have a collision-free exploration algorithm for graphs in $\mathcal{G}$ if a pebble is placed at one node and no pebble is placed at the other node.
\end{lemma}

\begin{proof}
    We show that in at least one graph in $\mathcal{G}$, if a pebble is placed at one node and no pebble is placed at the other node, then collision-free exploration is not possible to achieve. Without loss of generality, assume that a pebble is placed at $x$ and no pebble is placed at $y$. Let $\mathcal{A}$ be any exploration algorithm. If both agents make the same decision either $Nevermove()$ or $Move(i,t)$, irrespective of the pebble placement, then by proof of lemma \ref{lem:3_1}, it is impossible to have a collision-free exploration in $\mathcal{G}$.

    We consider the following cases where the agents take different decisions for different pebble placements. 

    \begin{itemize}
        \item[Case 1]If one agent takes decision $Nevermove()$ and the other agent takes decision $Move(i,t)$. Without loss of generality, assume that the agent at $x$ takes decision $Nevermove()$ and the agent at $y$ takes decision $Move(i,t)$. In this case, in graph $G_i$, the agent at $x$ stays at its initial node, while the agent at $y$ after waking up moves in its local round $t$ to $x$ where the other agent is staying. Hence, a collision occurs at $x$.
 
        \item[Case 2]If one agent takes decision $Move(i,t)$ and the other agent takes decision $Move(j,t')$ such that $i\ne j$ and $t\ne t'$. Without loss of generality, assume that the agent at $x$ takes decision $Move(i,t)$ and the agent at $y$ takes decision $Move(j,t')$ and $t < t'$. In this case, in graph $G_i$, if both agents wake up at the same time, then when the agent at $x$ moves to node $y$ in its local round $t$ where the other agent also has local round $t<t'$ and therefore stays at $y$ in local round $t$. Hence, a collision occurs at $y$.

        \item[Case 3]If one agent takes decision $Move(i,t)$ and the other agent takes decision $Move(j,t)$ such that $i\ne j$. Without loss of generality, assume that the agent at $x$ takes decision $Move(i,t)$ and the agent at $y$ takes decision $Move(j,t)$. In this case, in graph $G_i$, if the agent at $x$ wakes up at least one round before the agent at $y$, then when the agent at $x$ moves to node $y$ in its local round $t$ where the other agent has local round $<t$ and therefore stays at $y$. Hence, a collision occurs at $y$.

        \item[Case 4] If one agent takes decision $Move(i,t)$ and the other agent takes decision $Move(i,t')$ such that $t\ne t'$. Without loss of generality, assume that the agent at $x$ takes decision $Move(i,t)$ and the agent at $y$ takes decision $Move(j,t')$ and $t < t'$. In this case, in graph $G_i$, if both agents wake up at the same time, then when the agent at $x$ moves to node $y$ in its local round $t$ where the other agent also has local round $t<t'$ and therefore stays at $y$. Hence, a collision occurs at $y$.
    \end{itemize}
\end{proof}
  
From the above lemmas, we can give the following theorem.
\begin{theorem}
    There does not exist any algorithm $\mathcal{A}$ with some pebble placement strategies using which two identical agents, starting from $x$ and $y$, respectively, explore all the nodes of each graph in $\mathcal{G}$ without any collision.
\end{theorem}

\subsection{Algorithm for color index 3}\label{sec:4}

In this section, we provide an Algorithm for collision-free exploration that works for general graphs using pebble placement with color index 3.

\subsubsection{Preliminaries}

A breadth first search  (BFS) based traversal sequence (of infinite length) starting from a  node $v$ visits all the sequence of port numbers of length $i$ in the lexicographical order, for $i =1,2, \cdots$. We call this infinite sequence $AnonymousBFS(v)$. Consider the sequence $AnonymousBFS(x)$. Let $(f_1,f_2,\cdots,f_n)$ be the permutation on the set $V$ where $f_i \in V$ is the $i$-th node which is visited for the first time following $AnonymousBFS(x)$ starting from the node $x$. To be specific, let $S'$ be the minimum length subsequence of $AnonymousBFS(x)$ that visited $f_i$ starting from $x$. Then, starting from $x$, there exists exactly $n-i$ nodes in $V$ none of which is visited by $S'$. Define $z_x=f_n$. Let $y \ne x$ be any node in $V$.
 Let $P$ be the shortest path from $y$ to $z_x$ in $G$. Let $V_1=\{v_0(=y),v_1,\cdots,v_\ell (=z_x)\}$ be the set of nodes in path $P$. Let $V_2 = V \backslash V_1$ be the set of nodes in $G$ that are not in $P$. Let $p_0,p_1,\cdots,p_{\ell -1}$ be the outgoing port label on path $P$ such that node $v_j$ can be reached from node $v_{j-1}$ using port $p_{j-1}$. Let $P_1$ be the shortest path from $y$ to $x$ in $G$. Let $W_1=\{w_0(=y),w_1,\cdots,w_{\ell 1}(=x)\}$ be the set of nodes in path $P_1$. Let $V_3 = (V_2 \backslash W_1 )\cup \{x\}$ be the set of nodes in $G$ that are neither in $P$ nor in $P_1$. Define $z'_x$ as the last node in the ordered set $\{f_1, \cdots f_n\}\cap V_3$. Intuitively, $z'_x$ is the node in $V_3$ that is visited last by $AnonymousBFS(x)$ among all other nodes in $V_3$.

\subsubsection{High-level idea} 

Before describing the details of the pebble placement strategy and the corresponding exploration algorithm by the mobile agents, we give the high level idea of the overall algorithm.  In the high level idea, we described different phases of the overall algorithm and how to achieve them. However, while describing the details later, we have not mentioned these phases separately. The algorithm is presented using several subroutines which are executed by the agents based on the different colored pebbles they found in different stages of the algorithm. 

The overall goal is achieved in four different phases. The objectives of the phases are described below. 
\begin{itemize}
    \item[Phase 1] After waking up, both agents arrive at a distance 1 without any collision. Let these positions be $a_1$ and $a_2$. 
    \item[Phase 2] Each agent identifies whether the other agent is awake.
    \item[Phase 3] Agent at $a_2$ moves to the last visited node of the BFS traversal from $a_1$.
    \item[Phase 4] The agent at $a_1$ continues BFS traversal until it reaches the second last visited node in the BFS tree originating from $a_1$. 
\end{itemize}
 
If the first three steps can be achieved without collision, and step 4 is done only after the execution of step 3 is completed, then exploration of all the nodes of the graph can be guaranteed without any collision. 

In order to execute Phase 1, we identify one of the agents as the initiator which moves towards the other agent until they are one distance apart. During this activity, the other agent remains at its initial position irrespective of whether it has woken up or not. We identify this initiator agent based on their initial positions as follows.

Suppose that agents are initially placed at two different nodes $x'$ and $y'$ of the graph.  If $x'$ is present in all the shortest paths from $y'$ to $z_{x'}$, then $y'$ is called as $x$ and $x'$ is called as $y$. Otherwise, $x'$ is called $x$, and $y'$ is called $y$. We call the agent at $x$ as $A_1$ and the agent at $y$ as $A_2$. The agent $A_2$ is identified as the initiator and moves in Phase 1.

Since the agents are identical, we place a pebble of color red or blue (which color pebble is placed will be described later) at $y$ and no pebble at $x$. The agent, after waking up, sees if there is a pebble placed at its position and learns whether it is $A_1$ or $A_2$. To guide the agent $A_2$ towards $A_1$, some pebbles are placed at the nodes in the shortest path from $y$ to $x$. The agent $A_2$, after waking up, if it learns the presence of no agent at distance 1, explores all its neighbors and moves to the neighbor that contains a pebble. Continuing in this approach it reaches one step closer to $A_1$ and stops once it can sense the presence of $A_1$ for the first time. The agent $A_1$ remains at its initial position during the above activity until it identifies the presence of the other agent at distance 1.

To execute Phase 2, $A_2$ moves to a node that is at distance 2 from $A_1$. It waits for some time at this node and if $A_1$ does not arrive in its neighbor within this time, then $A_2$ again visits the node at distance 1 from $A_1$ and moves to a node that is at distance 2 from $A_1$ and waits there. $A_2$ continues this process until $A_1$ arrives at a node 1 distance apart from the node where $A_2$ waits and $A_2$ identifies that $A_1$ is awake. On the other hand, the agent $A_1$, after waking up, may learn the presence of the other agent at distance 1, but $A_2$ may be in sleep in case the initial distance between the agents is exactly 1. In this case, the agent $A_1$ remains at its position until $A_2$ leaves $A_1$'s neighborhood. As soon as this happens, $A_1$ knows that $A_2$ is awake. 

Once both agents identify that the other agent is awake, $A_1$ moves to node $x$ and waits there, while $A_2$ moves back to node $y$ and starts executing phase 3.

To execute Phase 3, we place blue/red pebbles at the nodes in the shortest path from $y$ to $z_x$. $A_2$ explores its neighbors and moves to a node with a red/blue pebble. $A_2$ continues this process until it finds a node with a red/blue pebble and none of its neighbors has any red/blue pebble. As per our convention, there exists a shortest path from $y$ to $z_x$ which does not contain $x$, where $A_1$ is present. This fact along with the placement of blue/red pebbles helps in completing Phase 3 in a collision-free manner. 

To execute Phase 4, a green pebble is placed at $z'_x$. $A_1$ waits at $x$ for $A_2$ to complete Phase 3 and then executes $AnonymousBFS(x)$ until it reaches the node with a green pebble. All the nodes in the path from $y$ to $z_x$ and in the paths from $x$ to $y$ are already visited by the two agents when Phase 3 is completed. As per the definition of $z'_x$, by the time the agent following the sequence of port numbers $AnonymousBFS(x)$ visits $z'_x$, all the nodes in $G$ that are not visited till Phase 3 are already visited. Hence, when $A_1$ visit the node with green pebble and stop, all the nodes in $G$ are explored.

Following are the challenges in implementing the above idea.
\begin{enumerate}
    \item How the agents know when a particular phase is ended. This knowledge is crucial as collision-free movements of the agents can be guaranteed if the activities of phase i start only after all the activities of phase i-1 are complete.
    \item How do the agents determine which agent initiates phase 1?
\end{enumerate}

The different colored pebbles and their placements help us to overcome the above challenge. The pebble placement algorithm uses three different color pebbles - \textit{green, red} and \textit{blue}. In the below subsection, we give the pebble placement strategy to overcome these challenges. Intuitively, the role of a pebble of each color is described below.

\textcolor{black}{The green color pebble at any node indicates that this is the final node where it has to move and terminate the algorithm. The red and blue color pebble is used by the agent to find a non-colliding shortest path to $z_x$. In order to guide the agent through the shortest path the red and blue color pebbles also indicate the order in which the neighbors should be explored so that collision is avoided. Red and blue color pebbles also identify which of the agents will start executing Phase 1. A combination of red/blue pebbles and conditions such as the degree of node, other agent in the neighborhood is used to avoid collision in corner cases. For example, if any of the starting nodes have degree $n-1$, then the red color pebble at different nodes indicates which port to skip to complete exploration without collision. }

\begin{figure}
    \centering
    \includegraphics[scale=1.0]{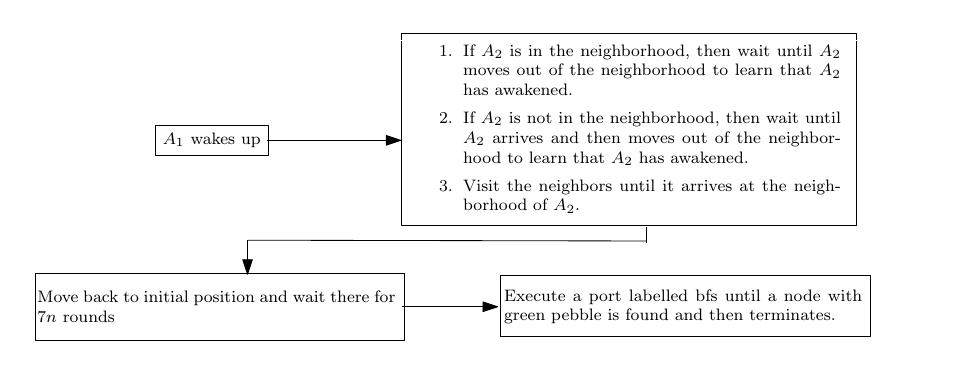}
    \caption{Flowchart showing the high level algorithm executed by agent $A_1$}
    \label{fig:flowchart1}
\end{figure}
\begin{figure}
    \centering
    \includegraphics[scale=1.0]{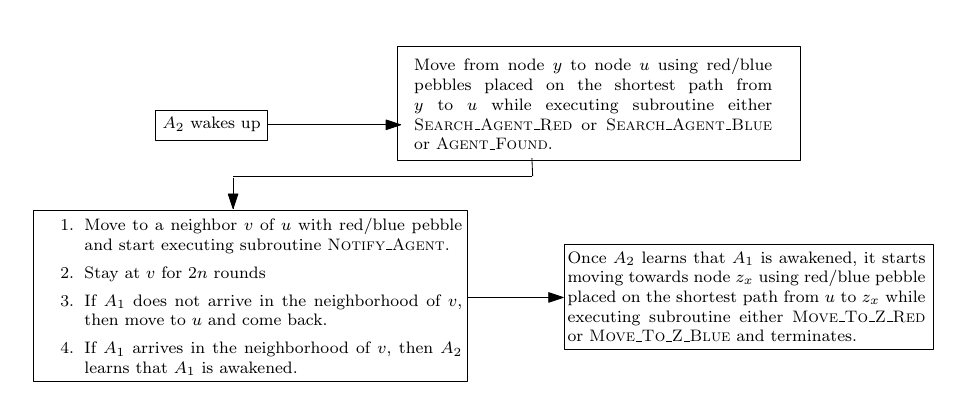}
    \caption{Flowchart showing the high level algorithm executed by agent $A_2$}
    \label{fig:flowchart2}
\end{figure}

Figure \ref{fig:flowchart1}, and Figure \ref{fig:flowchart2} show the flowchart of the proposed exploration algorithm indicating the major steps performed by the agents $A_1$ and $A_2$, respectively.

\subsubsection{Pebble Placement}\label{sec:4_1}

Suppose agents are initially placed at two different nodes $x'$ and $y'$ of the graph. Below we mention how to decide which agent moves in the first phase and which agent moves in the second phase.
\begin{itemize}
    \item If $x'$ is present on every shortest path from $y'$ to $z_{x'}$, then set $x=y'$ and $y=x'$. Else set $x=x'$ and $y=y'$.
\end{itemize}
The agent at $y$ executes phases 1, 2, and 3 and the agent at $x$ executes phases 2 and 4.


The pebbles are placed at different nodes of the graph as follows.

\begin{enumerate}
    \item If $deg(y) = n-1$, then a green colored pebble is placed at node $x$. Let $p$ be the port label at $y$ on edge $(x,y)$ and $y_i$ be the node that can be reached from $y$ using port $i$. The remaining pebbles are placed based on the following cases.
    \begin{itemize}
        \item If $p = 0$, then a red colored pebble is placed at node $y$ and a green colored pebble at node $y_{n-1}$.
        \item If $p = n-1$, then no pebble is placed at node $y$ and green colored pebble is placed at node $y_{n-2}$.
        \item If $ 0 < p < n-1$, then no pebble is placed at the node $y$ and a red colored pebble is placed at node $y_{p-1}$ and a green colored pebble at node $y_{n-1}$.
    \end{itemize}

    \item Suppose that $deg(y) \ne n-1$ and $deg(x) = n-1$, then a green colored pebble is placed at node $y$. Let $q$ be the port label at $x$ on edge $(x,y)$ and $x_i$ be the node that can be reached from $x$ using port $i$. Consider the following cases.
    \begin{itemize}
        \item If $q = 0$, then a red colored pebble is placed at node $x$ and a green colored pebble is placed at node $x_{n-1}$.
        \item If $q = n-1$, then no pebble at node $x$ and green colored pebble is placed at node $x_{n-2}$.
        \item If $ 0 < q < n-1$, then no pebble is placed at node $x$ and a red colored pebble is placed at node $x_{q-1}$ and a green colored pebble at node $x_{n-1}$.
    \end{itemize}

    \item If $deg(y) \ne n-1$ and $deg(x) \ne n-1$. Consider the following cases.
    \begin{enumerate}
        \item If $y=z_x$, then place a green-colored pebble at $z'_x$, and place a red-colored pebble at every node on $P_1\setminus \{x\}$.
        \item If $y \ne z_x$. Consider the following cases.
        \begin{enumerate}
            \item If $x$ is a neighbor of any node $v$ in ${P}$, then a red pebble is placed at node $z_x$. Consider any node $v_i\in V_1\setminus\{z_x\}$. If $v_i$ is not a neighbor of node $x$, then a blue pebble is placed at $v_i$. If $v_i$ is a neighbor of $x$ and the port label at $v_i$ on edge from $v_i$ to $v_{i+1}$ is smaller than the port label at $v_i$ on edge from $v_i$ to $x$, then a blue pebble is placed at $v_i$. Otherwise, a red pebble is placed at $v_i$. A green-colored pebble is placed at $z'_x$.
            \item Suppose that $x$ is not a neighbor of any node $v$ in ${P}$. Consider the following cases.
            \begin{enumerate}
                \item If ${P}\cap {P}_1=y$, then a red color pebble is placed at every node on ${P}_1\setminus \{x\}$ and a blue color pebble is placed at every node on ${P}\setminus \{y\}$. A green-colored pebble is placed at $z'_x$.
                \item If $|{P}\cap{P}_1|>1$. Let $U$ be the set of nodes that belong to both ${P}$ and ${P}_1$. Let $u\in U$ be a node such that $d(y,u) = \max\limits_{v\in U}\{ d(y,v) \}$. Let $P_{yu} = \{y, \cdots, u\}$ be the shortest path from $y$ to $u$. Let $P_{ux} = \{u, \cdots, x\}$ be the shortest path from $u$ to $x$. Let $P_{uz_x} = \{u, \cdots, z_x\}$ be the shortest path from $u$ to $z_x$. Then blue-colored pebble is placed at every node on $P_{yu}\setminus\{u\}$ and at every node on $P_{uz_x}\setminus \{u\}$. Red colored pebble is placed at every node on $P_{ux}\setminus \{x\}$. A green-colored pebble is placed at $z'_x$.
            \end{enumerate}
        \end{enumerate}
    \end{enumerate}
\end{enumerate}

If the two agents are not initially placed at the adjacent nodes, then the agents identify the initiator of phase 1 based on the degree of the nodes where they are initially placed, and whether there is a pebble at the initial position of the agents. If the degree of the initial node is at most $n-2$ and the agent does not find any pebble at its initial position, then it knows that the other agent will execute phase 1 and waits at its position for the other agent to complete phase 1. If the degree of the initial node is at most $n-2$ and the agent finds a red or blue-colored pebble at its initial position, then it knows that it will execute phase 1. Once both agents are at adjacent nodes at the end of phase 1, they start executing phase 2 to inform the other agent of its awakened state. 

When both agents are informed of their awakened state, the agent that has initiated Phase 1, executes Phase 3. While the other agent stays at its initial position and waits for the other agent to complete phase 3. The waiting agent waits for $7n$ rounds for the other agent to finish phase 3 before starting phase 4. Once the agent starts executing phase 4, it continues exploring using BFS traversal until it visits a node with a green pebble.

\subsubsection{The Exploration Algorithm}\label{sec:4_2}

The main idea behind the algorithm executed by the agents is to move an agent to a specific node in the graph without visiting the node where another agent is present. The agents follow Algorithm \ref{alg: init} and based on the presence of the pebble at its starting node $s$, it decides to start moving immediately or wait for some time.
According to the Algorithm \ref{alg: init}, if a green-color pebble is placed at $s$, then the agent stays at its current position and terminates. If no pebble is placed at $s$ and degree of $s$ is $n-1$, then the agent executes subroutine \textsc{One\_Level\_BFS}, while if no pebble is placed at $s$ and degree of $s$ is at most $n-2$, then agent executes subroutine \textsc{PortLabeledBFS}. If the degree of $s$ is $n-1$ and a red-color pebble is placed at $s$, then the agent executes subroutine \textsc{One\_Level\_BFS\_Color}. If the degree of $s$ is at most $n-2$, has a pebble (red or blue color) and the other agent is in the neighborhood of $s$, then it executes subroutine \textsc{Agent\_Found}. If the degree of $s$ is at most $n-2$ and the other agent is not in the neighborhood of $s$, then depending on the color of the pebble at $s$, it executes subroutine \textsc{Search\_Agent\_Red} if $s$ has red color pebble or subroutine \textsc{Search\_Agent\_Blue} if $s$ has a blue color pebble.

\begin{algorithm}
    \caption{\textsc{Initiate}($n$)}
    \label{alg: init}
     {
    {Let $s$ be the initial position of the agent.}\\
    \If{color of pebble at $s$ = green}
    {\label{init:1}
        {Stay at $s$ and terminate.}
    }\label{init:2}
    \ElseIf{$deg(s)=n-1$}
    {\label{init:3}
        \If{$s$ has a red color pebble}
        {
            {\textsc{One\_Level\_BFS\_Color}($s$)} 
        }
        \Else
        {
            {\textsc{One\_Level\_BFS}($s$,$0$)}
        }
    }\label{init:4}
    \Else
    {\label{init:5}
        \If{$s$ has a red or blue pebble}
        {\label{init:6}
            $color =$ color of pebble at $s$\\
            \If{another agent is in the neighborhood}
            {\label{init:7}
                {\textsc{Agent\_Found}($s,color,-1$)}\label{initline1}
            }\label{init:8}
            \Else
            {
                \If{$color = $ red}
                {\label{init:9}
                    $stack = []$\\
                    push$(stack,-1)$\\
                    {\textsc{Search\_Agent\_Red}($s,stack,-1$)}\label{initline2}
                }\label{init:10}
                \ElseIf{$color =$ blue}
                {\label{init:11}
                    {\textsc{Search\_Agent\_Blue}($s,-1$)}\label{initline3}
                }\label{init:12}
            }\label{init:13}
        }\label{init:14}
        \Else
        {\label{init:15}
            {\textsc{PortLabeledBFS}($s$)}
        }\label{init:16}
    }
    }
\end{algorithm}

Subroutine \textsc{One\_Level\_BFS} is executed when no pebble is found at $s$ and the degree of $s$ is $n-1$. According to this subroutine, the agent visits the neighbor of $s$ until it finds a node $u$ with a pebble. If the agent finds a red pebble at $u$, then it knows that the other agent is in the next node $u'$ and hence skips the port leading to node $u'$ and visits the node after $u'$. If the agent finds a green pebble at $u$, then it finds its final position and terminates there. If the agent finds a green pebble without finding a red pebble, then the agent knows that the other agent is in the node that can be reached from $s$ using port $n-1$.

\begin{algorithm}
    \caption{\textsc{One\_Level\_BFS}($s$,$p$)}
    \label{alg:onelevelbfs}
     {
    Move to adjacent node $u$ using port $p$.\\
    \If{$u$ has green pebble}
    {\label{bfslevel:1}
        Stay at $u$ and terminate.
    }\label{bfslevel:2}
    \ElseIf{$u$ has red pebble}
    {\label{bfslevel:3}
        $p = p+2$\\ \label{bfslevel:4}
        Move back to $s$.\\
        \textsc{One\_Level\_BFS}($s$,$p$)
    }\label{bfslevel:5}
    \Else
    {\label{bfslevel:6}
        $p = p+1$\\ \label{bfslevel:7}
        Move back to $s$.\\
        \textsc{One\_Level\_BFS}($s$,$p$)
    }\label{bfslevel:8}
    }
    \end{algorithm}

Subroutine \textsc{One\_Level\_BFS\_Color} is executed when a red-colored pebble is found at $s$ and the degree of $s$ is $n-1$. According to this subroutine, the agent visits the neighbor of $s$ until it finds a node $u$ with a green color pebble. If the color of the pebble at $s$ is red, then the agent knows that the other agent is at the node that can be reached using port 0. Hence, the agent starts traversing unexplored edges in increasing order of port label starting from 1.

\begin{algorithm}
    \caption{\textsc{One\_Level\_BFS\_Color}($s$)}
    \label{alg:onelevelbfscolor}
     {
    $p = 0$; \\ \label{bfscolor:1}
    \While{$p \le deg(s)-1$} 
    {\label{bfscolor:3}
        $p = p+1$;\\ \label{bfscolor:2}
        Move to adjacent node $u$ using port $p$.\\
        \If{$u$ has a green colored pebble}
        {
            Stay at $u$ and terminate.
        }
        \Else
        {
            Move back to node $s$.
        }
    }\label{bfscolor:4}
    }
\end{algorithm}

\begin{figure*}
     \centering
     \begin{subfigure}[b]{0.4\textwidth}
         \centering
         \includegraphics[width=.5\textwidth]{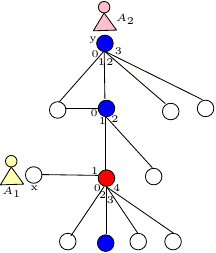}
         \caption{Initial Positions of the agents where $A_1$ is sleeping while $A_2$ is awake}
         \label{fig:a}
     \end{subfigure}
     \hfill
     \begin{subfigure}[b]{0.4\textwidth}
         \centering
         \includegraphics[width=.5\textwidth]{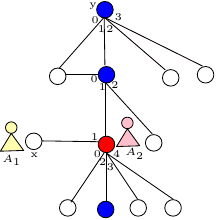}
         \caption{$A_2$ moves to $u$ using subroutine \textsc{Search\_Agent\_Blue} while $A_1$ is sleeping}
         \label{fig:b}
     \end{subfigure}
     \hfill
     \begin{subfigure}[b]{0.4\textwidth}
         \centering
         \includegraphics[width=.5\textwidth]{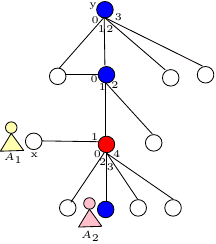}
         \caption{$A_2$ moves to $v$ using subroutine \textsc{Agent\_Found}}
         \label{fig:c}
     \end{subfigure}
     \hfill
     \begin{subfigure}[b]{0.4\textwidth}
         \centering
         \includegraphics[width=.5\textwidth]{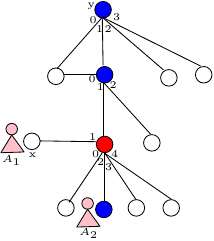}
         \caption{$A_2$ executes subroutine \textsc{Notify\_Agent} and wait for $2n$ rounds at $v$ and $A_1$ wakes up during that time}
         \label{fig:d}
     \end{subfigure}
     \hfill
     \begin{subfigure}[b]{0.4\textwidth}
         \centering
         \includegraphics[width=.5\textwidth]{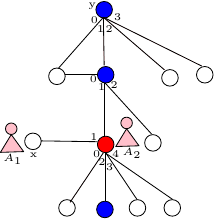}
         \caption{$A_2$ executes subroutine \textsc{Notify\_Agent} and moves back to $u$ after $2n$ rounds}
         \label{fig:e}
     \end{subfigure}
     \hfill
     \begin{subfigure}[b]{0.4\textwidth}
         \centering
         \includegraphics[width=.5\textwidth]{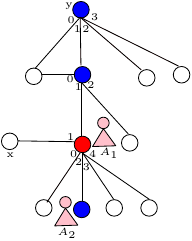}
         \caption{$A_2$ executes subroutine \textsc{Notify\_Agent} and wait for $2n$ rounds at $v$ while $A_1$ moves to $u$}
         \label{fig:f}
     \end{subfigure}
     \caption{Showing the case where $A_1$ is not awake when $A_2$ visits $u$. The placements of blue and red pebbles of the path from $y$ to $x$ are shown. The blue (red) node indicates the placement of a blue (red) pebble. The white nodes are the nodes where no pebbles are placed.}
     \label{fig:notify1}
\end{figure*}

\begin{figure*}
     \centering
     \begin{subfigure}[b]{0.45\textwidth}
         \centering
         \includegraphics[width=.5\textwidth]{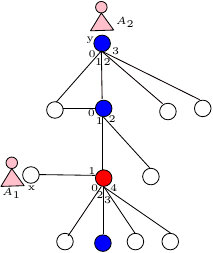}
         \caption{Initial Positions of the agents where both agents $A_1$ and $A_2$ are awake}
         \label{fig:2a}
     \end{subfigure}
     \hfill
     \begin{subfigure}[b]{0.45\textwidth}
         \centering
         \includegraphics[width=.5\textwidth]{figures/notify_agent/inform_agent.pdf}
         \caption{$A_2$ moves to $u$ using subroutine \textsc{Search\_Agent\_Blue} and $A_1$ is awake}
         \label{fig:2b}
     \end{subfigure}
     \hfill
     \begin{subfigure}[b]{0.45\textwidth}
         \centering
         \includegraphics[width=.5\textwidth]{figures/notify_agent/wait_at_v.pdf}
         \caption{$A_2$ moves to $v$ using subroutine \textsc{Agent\_Found}}
         \label{fig:2c}
     \end{subfigure}
     \hfill
     \begin{subfigure}[b]{0.45\textwidth}
         \centering
         \includegraphics[width=.5\textwidth]{figures/notify_agent/both_awake.pdf}
         \caption{$A_2$ executes subroutine \textsc{Notify\_Agent} and wait for $2n$ rounds at $v$ while $A_1$ moves to $u$}
         \label{fig:2d}
     \end{subfigure}
     \caption{Showing the case where $A_1$ is awake when $A_2$ visits $u$. The placements of blue and red pebbles of the path from $y$ to $x$ are shown.}
     \label{fig:notify2}
\end{figure*}

Agent executes subroutine \textsc{Agent\_Found} if the degree of node $s$ (initial position of the agent) is at most $n-2$ and $x$ is a neighbor of at least one node $v\in P$. According to steps \ref{agentfound:1}-\ref{agentfound:2} of this subroutine, if the agent finds a red-colored pebble at $v$, it visits the neighbors of $v$ in decreasing order of the port number, while if it finds a blue-colored pebble at $v$, then it visits the neighbor of $v$ in increasing order of the port number until it finds a node $u$ with a red or blue-colored pebble. According to steps \ref{agentfound:3}-\ref{agentfound:4}, once it finds a node $u$ with a red or blue pebble, it checks if the other agent is in its neighborhood or not. If it finds the other agent is in the neighborhood of $u$, then it calls subroutine \textsc{Agent\_Found}($u$) and continues to search for a node with a pebble. According to steps \ref{agentfound:5}-\ref{agentfound:6}, if the agent does not find the other agent in the neighborhood of $u$, then it calls subroutine \textsc{Notify\_Agent} to inform the other agent about its awakened state. If the color of the pebble at node $u$ is blue, then once the agent completes the execution of subroutine \textsc{Notify\_Agent}, it executes subroutine \textsc{Move\_to\_Z\_Blue}.

\begin{algorithm}
    \caption{\textsc{Agent\_Found}($v,color,e$)}
    \label{alg:agentfound}
     {
    $i = 0$; $pebble\_f = 0;$\\
    \While{$pebble\_f\ne 1$}
    {\label{agentfound:1}
        \If{$color = $ red}
        {
            $p = deg(v) - 1 - i$;
        }
        \ElseIf{$color = $ blue}
        {
            $p = i$;
        }
        \If{$p\ne e$}
        {
            Move to adjacent node $u$ using port $p$.\\
            \If{$u$ has red or blue colored pebble}
            {
            $color\_u =$ color of pebble at $u$.\\
            $e_1 = $ Incoming port at $u$.\\
            $pebble\_f = 1$;
            }
            \Else
            {
                Move back to $v$.\\
                $i = i+1$;
            }
        }
    }\label{agentfound:2}
    \If{another agent is in the neighborhood}
    {\label{agentfound:3}
        \textsc{Agent\_Found}($u,color\_u,e_1$)
    }\label{agentfound:4}
    \Else
    {\label{agentfound:5}
        \textsc{Notify\_Agent}$(u,e_1)$\\ \label{agentfound:7}
        \If{$color\_u = $ blue}
        {
            \textsc{Move\_to\_Z\_Blue}$(u,e_1)$ \label{agentfound:8}
        }
    }\label{agentfound:6}
    }
\end{algorithm}

Agent executes subroutine \textsc{Search\_Agent\_Red} if the degree of node $s$ (initial position of the agent) is at most $n-2$ and the pebble placed at $s$ is red and $x$ is not a neighbor of any node $v \in P$. According to steps \ref{searchagentred:1}-\ref{searchagentred:2} of this subroutine, if the agent finds a red-colored pebble at $v$, it visits the neighbors of $v$ in decreasing order of the port number until it finds a node $u$ with a red-colored pebble. According to steps \ref{searchagentred:3}-\ref{searchagentred:4}, once it finds a node $u$ with a red pebble, it checks if the other agent is in its neighborhood or not. If it finds the other agent is in the neighborhood of $u$, then it returns to $v$ and calls subroutine \textsc{Notify\_Agent} to inform the other agent about its awakened state. Once the agent completes the execution of subroutine \textsc{Notify\_Agent}, it executes subroutine \textsc{Return\_Back}. According to steps \ref{searchagentred:5}-\ref{searchagentred:6}, if the agent does not find the other agent in the neighborhood of $u$, then it calls subroutine \textsc{Search\_Agent\_Red} and continues to search for a node with a pebble. 

\begin{algorithm}
    \caption{\textsc{Search\_Agent\_Red}$(v,stack,e)$}
    \label{alg:searchagentred}
    {
    $p=deg(v)-1$;
    $r\_pebble\_f=0$;\\
    \While{$r\_pebble\_f\ne 1$} 
    {\label{searchagentred:1}
        \If{$p\ne e$}
        {
            Move to adjacent node $u$ using port $p$\\
            \If{$u$ has a red colored pebble}
            {
                $r\_pebble\_f=1$;\\
                $e_1 = $ Incoming port at $u$\\
            }
            \Else
            {
                {Move back to $v$}\\
                {$p=p-1$;}\\
            }
        }
    }\label{searchagentred:2}
    \If{another agent is in the neighborhood}
    {\label{searchagentred:3}
        Move back to $v$\\
        \textsc{Notify\_Agent}$(v,p)$\\ \label{searchagentred:7}
        \textsc{Return\_Back}$(stack)$\\ \label{searchagentred:8}
    }\label{searchagentred:4}
    \Else
    {\label{searchagentred:5}
        push$(stack,e_1)$\\
        \textsc{Search\_Agent\_Red}$(u,stack,e_1)$
    }\label{searchagentred:6}
    }
    \end{algorithm}

Agent executes subroutine \textsc{Search\_Agent\_Blue} if the degree of node $s$ (initial position of the agent) is at most $n-2$ and it finds a blue-colored pebble at a node $v$ and the other agent is not in the neighborhood of $v$. This means that either $x$ is a neighbor of at least one node on the shortest path from $y$ to $z_x$ or $x$ is not a neighbor of any node on the shortest path from $y$ to $z_x$ (final position of the agent) but both the shortest path from $y$ to $x$ and $y$ to $z_x$ share nodes apart from $y$. According to steps \ref{searchagentblue:1}-\ref{searchagentblue:2} of this subroutine, if the agent finds a blue-colored pebble at $v$, it visits the neighbors of $v$ in increasing order of the port number until it finds a node $u$ with a red or blue-colored pebble. According to steps \ref{searchagentblue:3}-\ref{searchagentblue:4}, once it finds a node $u$ with a red or blue pebble, it checks if the other agent is in its neighborhood or not. If it finds the other agent is in the neighborhood of $u$, then it means that $x$ is a neighbor of at least node $v\in P\setminus\{y,z_x\}$ and the agent calls subroutine \textsc{Agent\_Found}. According to steps \ref{searchagentblue:5}-\ref{searchagentblue:6}, if the agent does not find the other agent in the neighborhood of $u$, then it checks the color of the pebble at $u$. If $u$ has a blue color pebble, then it executes subroutine \textsc{Search\_Agent\_Blue} while if $u$ has a red color pebble, then it executes subroutine \textsc{Search\_Agent\_Red} to search for a node with a pebble.

The agent calls subroutine \textsc{Move\_to\_Z\_Blue} after it has informed the other agent about its awakened state and knows that the node it is currently at, is not the final node $z_x$. Subroutine \textsc{Move\_to\_Z\_Blue} is called when the agent is at a node with a blue-colored pebble and has the knowledge that the other agent was a neighbor of at least one node on the shortest path from $y$ to $z_x$. According to steps \ref{finalmoveblue:1}-\ref{finalmoveblue:2} of subroutine \textsc{Move\_to\_Z\_Blue}, the agent visits its neighborhood in increasing order of the port number until it finds a node $u$ with a red or blue colored pebble. According to steps \ref{finalmoveblue:3}-\ref{finalmoveblue:4}, if it finds a blue-colored pebble, then it calls subroutine \textsc{Move\_to\_Z\_Blue} and continues to search in the neighborhood of $u$ for a node with a red or blue-colored pebble. According to steps \ref{finalmoveblue:5}-\ref{finalmoveblue:6}, if it finds a red-colored pebble and the other agent is not in its neighborhood, then node $u = z_x$ and the agent has reached its final position. Hence, the agent stops at $u$ and terminates.

\begin{algorithm}
\caption{\textsc{Search\_Agent\_Blue}$(v,e)$}
    \label{alg:searchagentblue}
    \begin{multicols}{2}
     {
    {$p = 0$; $b\_pebble\_f = 0$;}\\
    \While{$b\_pebble\_f \ne 1$}
    {\label{searchagentblue:1}
        \If{$p\ne e$}
        {
            Move to adjacent node $u$ using port $p$.\\
            \If{$u$ has a red or blue color pebble}
            {
                $color = $ color of pebble at node $u$.\\
                $e_1 = $ Incoming port at $u$.\\
                $b\_pebble\_f = 1$;
            }
            \Else
            {
                Move back to $v$.\\
                $p = p+1$;
            }
        }
    }\label{searchagentblue:2}
    \If{another agent is in the neighborhood}
    {\label{searchagentblue:3}
        \textsc{Agent\_Found}($u, color,e_1$)\label{blue1}
    }\label{searchagentblue:4}
    \Else
    {\label{searchagentblue:5}
        \If{$color = $ blue}
        {\label{searchagentblue:7}
            \textsc{Search\_Agent\_Blue}$(u,e_1)$
        }\label{searchagentblue:8}
        \Else
        {\label{searchagentblue:9}
            $stack = []$\\
            push$(stack,e_1)$\\
            \textsc{Search\_Agent\_Red}$(u,stack,e_1)$\label{blue2}
        }\label{searchagentblue:10}
    }\label{searchagentblue:6} 
    }
    \end{multicols}
\end{algorithm}

The agent calls subroutine \textsc{Move\_to\_Z\_Red} after it has informed the other agent about its awakened state and the node it is currently at, either in its final position or is not in its final position. Subroutine \textsc{Move\_to\_Z\_Red} is called when the agent is at a node with a red-colored pebble and has the knowledge that the other agent was not a neighbor of any node on the shortest path from $y$ to $z_x$. According to steps \ref{finalmovered:1}-\ref{finalmovered:2} of subroutine \textsc{Move\_to\_Z\_Red}, the agent visits its neighborhood in increasing order of the port number until it finds a node $u$ with a blue-colored pebble.  According to steps \ref{finalmovered:3}-\ref{finalmovered:4}, if it finds a blue-colored pebble, then it calls subroutine \textsc{Move\_to\_Z\_Red} and continues to search in the neighborhood of $u$ for a node with a blue-colored pebble while according to steps \ref{finalmovered:5}-\ref{finalmovered:6}, if it does not find any blue-colored pebble in the neighborhood of $u$, then node $u = z_x$ and the agent has reached its final position. Hence agent stops at $u$ and terminates.

Subroutine \textsc{Notify\_Agent} is used by the moving agent to inform the other agent that it is awake. Subroutine \textsc{Notify\_Agent} is executed when this agent arrives at a node $u$ and finds the other agent in its neighborhood. According to subroutine \textsc{Notify\_Agent}, the agent either waits for $2n$ rounds at its current position (node $v$) or for the other agent to arrive in its neighborhood, whichever occurs first. If another agent doesn't arrive in the neighborhood within $2n$ rounds, then the agent moves to a node $u$ using port $p$ and informs the other agent that it has awakened and returns back to its position in the next round. The agent continues this cycle of waiting and visiting until the other agent arrives in its neighborhood.

Once the agent gets the information that the other agent is awake, then it executes subroutine \textsc{Return\_Back}. According to subroutine \textsc{Return\_Back}, using port labels stored in the stack, the agent returns to its initial node and executes subroutine \textsc{Move\_to\_Z\_Red}.

In Algorithm \textsc{PortLabeledBFS}, the agent does not consider red or blue color pebbles. Since no pebble was found at the initial position $s$, the agent waits at its initial position for the other agent to arrive and then depart from its neighborhood so that it can learn that the other agent is awakened. Once the other agent has moved out of its neighborhood, it visits all of its neighbors until it reaches a node $u$ such that the other agent is in the neighborhood of $u$.

\begin{algorithm}
    \caption{\textsc{Move\_to\_Z\_Blue}$(v,e)$}
    \label{alg:finalmoveblue}
     {
    {$p = 0$; $b\_pebble\_f = 0$}\\
    \While{$b\_pebble\_f \ne 1$ or $b\_pebble\_f \ne 2$}
    {\label{finalmoveblue:1}
        \If{$p\ne e$}
        {
            Move to adjacent node $u$ using port $p$\\
            \If{$u$ has a blue colored pebble}
            {
                $e_1 = $ Incoming port at $u$.\\
                $b\_pebble\_f = 1$;
            }
            \ElseIf{$u$ has a red colored pebble}
            {
                $b\_pebble\_f = 2$;
            }
            \Else
            {
                Move back to $v$.\\
                $p = p+1$;
            }
        }
    }\label{finalmoveblue:2}
    \If{$b\_pebble\_f = 1$}
    {\label{finalmoveblue:3}
        \textsc{Move\_to\_Z\_Blue}$(u,e_1)$
    }\label{finalmoveblue:4}
    \ElseIf{$b\_pebble\_f = 2$}
    {\label{finalmoveblue:5}
        Stop at $u$ and terminate.
    }\label{finalmoveblue:6}
    }
\end{algorithm}

 This action will inform the other agent that this agent has also awakened. The agent then returns to $s$ and waits for $7n$ rounds (which is the maximum time an agent can take to reach from $y$ to $z_x$) at $s$. It then visits all the sequence of port numbers of length $i$ in the lexicographical order, for $i =1,2, \cdots$ until it finds a node with a green pebble.

\begin{algorithm}
    \caption{\textsc{Move\_to\_Z\_Red}$(v,e)$}
    \label{alg:finalmovered}
    {
    $d = $ degree of node $v$\\
    {$p = 0$; $b\_pebble\_f = 0$}\\
    \While{$b\_pebble\_f \ne 1$ or $p<d$}
    {\label{finalmovered:1}
        \If{$p\ne e$}
        {
            Move to adjacent node $u$ using port $p$\\
            \If{$u$ has a blue colored pebble}
            {
                $e_1 = $ Incoming port at $u$.\\
                $b\_pebble\_f = 1$;
            }
            \Else
            {
                Move back to $v$.\\
                $p = p+1$;\\
            }
        }
    }\label{finalmovered:2}
    \If{$p = d$}
    {\label{finalmovered:3}
        Stay at $v$ and terminate.
    }\label{finalmovered:4}
    \Else
    {\label{finalmovered:5}
        \textsc{Move\_to\_Z\_Red}$(u,e_1)$
    }\label{finalmovered:6}
    }
\end{algorithm}

\begin{algorithm}
    \caption{\textsc{Notify\_Agent}($v$,$p$)}
    \label{alg:notifyAgent}
    {
    $flag = 0$\\
    \For{$i = 1$ to $2n$}
    {\label{notifyagent:1}
        \If{another agent is in the neighborhood of $v$}
        {\label{notifyagent:3}
            $flag = 1$\\
            break
        }\label{notifyagent:4}
        \Else
        {\label{notifyagent:5}
            Wait at $v$
        }\label{notifyagent:6}
    }\label{notifyagent:2}
    \If{$flag = 0$}
    {\label{notifyagent:7}
        Move to adjacent node $u$ using port $p$\\
        Move back to $v$\\
        \textsc{Notify\_Agent}($v$,$p$)
    }\label{notifyagent:8}
    \Else
    {
        return
    }
    }
\end{algorithm}

\begin{algorithm}
    \caption{\textsc{Return\_Back}(stack)}
    \label{alg:returnBack}
     {
    \If{$len(stack) > 1$}
    {
        $e=$pop(stack)\\
        Move to node $u$ using port $e$\\
        \textsc{Return\_Back}(stack)
    }
    \Else
    {
        Let $v$ be the current position of the agent.\\
        $e=$pop(stack)\\
        \textsc{Move\_to\_Z\_Red}$(v,e)$
    }
    }
\end{algorithm}


\begin{algorithm}
    \caption{\textsc{PortLabeledBFS}($s$)}
    \label{alg:bfs}
     {
    \While{1}
    {\label{bfs:1}
        \If{another agent is in its neighborhood}
        {\label{bfs:3}
            Wait for $2$ rounds at $s$\\ \label{bfsline1}
            \If{another agent is in the neighborhood}
            {\label{bfs:5}
                Go to line \ref{bfsline1}
            }\label{bfs:6}
            \Else
            {\label{bfs:7}
                \For{$s$ has unexplored edges}
                {
                    $p = $ lowest port label among unexplored edges\\
                    Move to adjacent node $u$ using port $p$\\
                    \If{another agent is in its neighborhood}
                    {
                        Move back to $s$ and mark all edges as unexplored\\
                        Go to line \ref{bfsline2}
                    }
                    \Else
                    {
                        Move back to $s$
                    }
                }
            }\label{bfs:8}
        }
        \Else
        {\label{bfs:4}
            Wait at $s$
        }
    }\label{bfs:2}
    {Wait for $7n$ rounds.}\\ \label{bfsline2}
    {$i=1$}\\
    \While{a green pebble is not found}
    {
        {Explore all sequences of port numbers of length $i$}\\
        {come back to $s$}\\
        {$i=i+1$}
    }
    }
\end{algorithm}

\textcolor{black}{Figure \ref{fig:notify1} and Figure \ref{fig:notify2} show the execution of the subroutine \textsc{Search\_Agent\_Blue} to move agent $A_2$ in the neighborhood of $A_1$, subroutine \textsc{Agent\_Found} to move agent $A_2$ out of the neighborhood of $A_1$, and \textsc{Notify\_Agent} to find if $A_1$ is awake. The agent filled with yellow indicates that it is asleep in the current round while the agent filled with pink indicates that it is awake. }

\subsubsection{Correctness}

In this section, we show that with the pebble placement described in section \ref{sec:4_1} and the corresponding algorithm given in section \ref{sec:4_2}, the two agents visit all nodes of the graph and no collision occurs between two agents during the execution of the exploration algorithm. Let $A_1$ and $A_2$ be the agents deployed at nodes $x$ and $y$ respectively.

The following lemma is useful in establishing the fact that for at least one agent, its movement path towards the last node in the BFS traversal rooted at the starting node of the other agent does not contain the starting node of the other agent.

\begin{lemma}\label{lem:4_1}
    Consider the traversal sequences $AnonymousBFS(x)$ and $AnonymousBFS(y)$  starting from node $x$ and node $y$, respectively. If $x$ is present on the shortest path from node $y$ to $z_x$, then $y$ can not be present on the shortest path from node $x$ to  $z_y$.
\end{lemma}

\begin{proof} We prove this lemma by contradiction.
    Suppose that $x$ is a node on the shortest path from node $y$ to node $z_x$ and $y$ is a node on the shortest path from node $x$ to node $z_y$. Let $d(a,b)$ denote the shortest distance from node $a$ to $b$. Since $x$ and $y$ are two different nodes where agents are initially placed, therefore $d(x,y)\ge 1$.
    If $x$ is a node on shortest path from node $y$ to the node $z_x$, then
    \begin{equation}\label{eq:1}
        d(y,z_x) = d(y,x) + d(x,z_x)
    \end{equation}
    If $y$ is a node on shortest path from node $x$ to node $z_y$, then
    \begin{equation}\label{eq:2}
        d(x,z_y) = d(x,y) + d(y,z_y)
    \end{equation}

    We first show that $z_x\ne z_y$. To show this, suppose that $d(z_x,z_y)=0$. Then, $d(x,z_y)=d(x,z_x)$ and $d(y,z_y)=d(y,z_x)$. From equation \ref{eq:2}, we get
    \begin{equation}
        d(x,z_x) = d(x,y)+d(y,z_x) \label{eq:3}
    \end{equation}
    Substituting equation \ref{eq:3} in equation \ref{eq:1}, we get $d(y,z_x) = d(x,y)+d(x,y)+d(y,z_x)$. Therefore, $d(x,y) = 0$, which contradicts the fact that $x \ne y$. Therefore, $z_x$ and $z_y$ are two different nodes.\\
    Now, from Equation\ref{eq:2}, we get:
    \begin{equation}\label{eq:4}
        d(y,z_y) = d(x,z_y)-d(x,y)
    \end{equation}
    Since $z_y$ is the node visited last in $AnonymousBFS(y)$ so, $d(y,z_y)\ge d(y,z_x)$. Hence, from equations \ref{eq:1} and \ref{eq:4}, we get
    
     $d(x,z_y)-d(x,y) \ge d(y,x)+d(x,z_x)$, which implies that 
     $d(x,z_y) \ge d(x,z_x)+2d(x,y)$.
    Since $d(x,y)\ge 1$, $d(x,z_y)>d(x,z_x)$ which is not true as $z_x$ is the node visited last in $AnonumousBFS(x)$. Hence, $y$ can not be present on the shortest path from node $x$ to node $z_y$ if $x$ is a node on the shortest path from $y$ to $z_x$.
\end{proof}

The following lemma gives the correctness of the algorithm in case of either $deg(y) = n-1$ or $deg(x) = n-1$.

\begin{lemma}\label{lem:4_3}
    If either $deg(y) = n-1$ or $deg(x) = n-1$, then by following Algorithm \ref{alg: init}, each node in $G$ is visited by at least one agent and there is no collision between the agents during the execution of Algorithm \ref{alg: init}.
\end{lemma}
    
\begin{proof}
    We prove that if $deg(y)=n-1$, then the statement of the lemma is true. Similar proof will work when $deg(x)=n-1$.
 
    Consider the following case.
    \begin{enumerate}
        \item[Case 1] The agent at $y$ finds a red pebble at $y$. According to the pebble placement strategy explained in Section \ref{sec:4_1}, a red pebble is placed at $y$ only when the port number of the edge $(y,x)$ at $y$ is 0. According to Algorithm \ref{alg: init}, if $deg(y) = n-1$ and a red pebble is placed at $y$, then the agent calls subroutine \textsc{One\_Level\_BFS\_Color} according to which in line \ref{bfscolor:1}, the agent $A_2$ initializes the port number $p$ as 0 and in line \ref{bfscolor:2}, updates the port number value to $p+1$ before moving from node $y$. According to steps \ref{bfscolor:3}-\ref{bfscolor:4}, the agent moves to the neighbor node and back to $y$ until $p\le n-1$. Hence, the agent $A_2$ visits all the neighbor nodes of $y$ except node $x$ that can be reached from $y$ using port 0. On the other hand, a green pebble is placed at $x$ and the agent at $x$ once wakes up, sees this green pebble and immediately terminates (lines \ref{init:1}-\ref{init:2} of Algorithm \ref{alg: init}). Hence, there is no collision between the agents while executing algorithm \ref{alg: init}.
        
        \item[Case 2] The agent finds no pebble at $y$. According to Algorithm \ref{alg: init}, if $deg(y) = n-1$ and a no pebble is placed at $y$, then $A_2$ calls subroutine \textsc{One\_Level\_BFS}, according to which in steps \ref{bfslevel:6}-\ref{bfslevel:8}, $A_2$ visits neighbor nodes in increasing order of port numbers. On the other hand, a green pebble is placed at $x$ and the agent at $x$ once wakes up, sees this green pebble and immediately terminates (lines \ref{init:1}-\ref{init:2} of Algorithm \ref{alg: init}). Let $p$ be the port label at node $y$ on edge $(x,y)$. Consider the following case.
        \begin{enumerate}
            \item If $A_2$ reaches node $y_{n-2}$ with a green pebble without finding a red pebble, then according to the pebble placement strategy given in section \ref{sec:4_1}, node $x$ can be reached from node $y$ using port $n-1$. According to steps \ref{bfslevel:1}-\ref{bfslevel:2} of subroutine \textsc{One\_Level\_BFS}, $A_2$ stops and terminates the algorithm as soon as it finds a node with a green pebble. Hence, $A_2$ visits all its neighbor nodes except node $x$. Therefore, there is no collision between the agents while executing algorithm \ref{alg: init}.

            \item If $A_2$ finds a red pebble at node $y_{p-1}$. According to the pebble placement strategy, node $x$ can be reached from node $y$ using port $p$. According to steps \ref{bfslevel:3}-\ref{bfslevel:5} of subroutine \textsc{One\_Level\_BFS}, $A_2$ updates the value of port number as $p+1$ and skips the port $p$ where the other agent can be found. Hence, $A_2$ visits all its neighbor nodes except node $x$. Hence, there is no collision between the agents while executing algorithm \ref{alg: init}.
        \end{enumerate}
    \end{enumerate}
\end{proof}

We define a node $u$ as follows.
\begin{itemize}
    \item If $x$ is neighbor to some nodes on the path $P$, then we define $u$ as the node on $P$ 
 that is a neighbor of $x$ and closest to $z_x$ among all other neighbors of $x$ on $P$. 
    \item If $x$ is not a neighbor to any node on the path $P$, then we define $u$ as the node on the path $P_1$ (the shortest path from $y$ to $x$) that is neighbor node of $x$.
\end{itemize}

The next three lemmas ensure two important aspects of our algorithm: (1) $A_2$ moves to the one-hop neighborhood of $A_1$, and identifies whether $A_1$ is awake or not, and (2) if it is not awake yet, waits till it becomes wake up, and then moves to $z_x$ collision-free. 

\begin{lemma}\label{lem:4_2}
    If $deg(y) < n-1$ and $deg(x) < n-1$, then $A_2$ successfully moves from $y$ to $u$ without any collision.
\end{lemma}
    
\begin{proof}
    When $A_2$ wakes up at $y$ and finds that $deg(y) < n-1$, according to the pebble placement strategy, $A_2$ either finds a red or a blue-colored pebble at $y$. According to steps \ref{init:5}-\ref{init:14} of Algorithm \ref{alg: init}, if $A_2$ finds $deg(y) < n-1$ and a red or blue pebble at $y$, then $A_2$ checks if $A_1$ is in its neighborhood or not. Consider the following cases.
    \begin{enumerate}
        \item\label{case 3a} If $A_1$ is in the neighborhood of $A_2$ when it wakes up, then $d(x,y) = 1$. In this case, $A_2$ calls subroutine \textsc{Agent\_Found} (according to steps \ref{init:7}-\ref{init:8} of Algorithm \ref{alg: init}). According to steps \ref{agentfound:1}-\ref{agentfound:2} of \textsc{Agent\_Found}, $A_2$ visits its neighbor in increasing port numbers if $A_2$ finds a blue pebble. Otherwise, $A_2$ visits its neighbor in decreasing port numbers if $A_2$ finds a red pebble until it finds a node $v_1$ with a red or blue pebble. If $A_1$ is not in the neighborhood of $v_1$, then $u=y$.
 
        If $A_1$ is in the neighborhood of $v_1$, then $A_2$ calls subroutine \textsc{Agent\_Found} from $v_1$ (lines \ref{agentfound:3}-\ref{agentfound:4}) and visits its neighbor in increasing or decreasing port numbers depending on the color (blue or red) of the pebble at $v$ until it finds node $v_2$ with a red or blue pebble. If $A_1$ is not in the neighborhood of $v_2$, then $u=v_1$.

        If $A_1$ is in the neighborhood of $v_2$, then since $x$ can be a neighbor of at most three nodes on $P$, $u = v_2$.

        According to the pebble placement strategy, if the port label on edge connecting nodes $y$ and $v_1$ at $y$ is greater than the port label on edge connecting nodes $y$ and $x$ at $y$, then a red pebble is placed at node $y$ on $P$. Otherwise, a blue pebble is placed at $y$. Similarly, if the port label on edge connecting nodes $v_1$ and $v_2$ at $v_1$ is greater than the port label on edge connecting nodes $v_1$ and $x$ at $v_1$, then a red pebble is placed at node $v_1$ on $P$. Otherwise, a blue pebble is placed at $v_1$. Hence, when $A_2$ visits neighbors of $y$ and $v_1$, then it finds the node with a red or blue pebble before $x$ and no collision happens between the two agents when $A_2$ moves from $y$ to $u$.

        \item\label{case 3b} If $A_1$ is not in the neighborhood of $A_2$ and $A_2$ finds a red pebble at $y$, then $A_2$ knows that $d(x,y) \ne 1$ and $x$ is not a neighbor of any node on path $P$. In this case, $A_2$ calls subroutine \textsc{Search\_Agent\_Red} (according to lines \ref{init:9}-\ref{init:10} of Algorithm \ref{alg: init}). According to lines \ref{searchagentred:1}-\ref{searchagentred:2} of subroutine \textsc{Search\_Agent\_Red}, $A_2$ visits its neighbor in decreasing order of port numbers until it finds a node $v$ with a red pebble. According to lines \ref{searchagentred:5}-\ref{searchagentred:6}, $A_2$ checks if the other agent is in its neighborhood. If $A_1$ is not in the neighborhood of $v$, then $A_2$ calls subroutine \textsc{Search\_Agent\_Red} and repeats the process. If $A_1$ is in the neighborhood of $v$, then $u = v$.

        Since no node before $u$ is the neighbor of $x$, then no collision happens between the two agents when $A_2$ moves from $y$ to $u$.

        \item If $A_1$ is not in the neighborhood of $A_2$ and $A_2$ finds a blue pebble at $y$, then $A_2$ learns that $d(x,y) \ne 1$ and $x$ may be a neighbor of some node on path $P$. In this case, $A_2$ calls subroutine \textsc{Search\_Agent\_Blue} (according to lines \ref{init:11}-\ref{init:12} of Algorithm \ref{alg: init}). According to lines \ref{searchagentblue:1}-\ref{searchagentblue:2} of subroutine \textsc{Search\_Agent\_Blue}, $A_2$ visits its neighbor in increasing order of port numbers until it finds a node $u'$ with a red or blue pebble. According to lines \ref{searchagentblue:3}-\ref{searchagentblue:6}, $A_2$ checks if the other agent is in its neighborhood. Consider the following cases.
        \begin{enumerate}
            \item \label{case 3c} If $A_1$ is in the neighborhood, then $x$ is a neighbor of at least one node on $P$, and $A_2$ calls subroutine \textsc{Agent\_Found} from $u'$ (according to lines \ref{searchagentblue:3}-\ref{searchagentblue:4}). Using similar arguments as in Case \ref{case 3a} above, we can say that no collision happens between the two agents when $A_2$ moves from $y$ to $u$.

            \item \label{case 3d} If $A_1$ is not in the neighborhood and $A_2$ finds a red pebble, then $x$ is not a neighbor of any node on $P$, and $A_2$ calls subroutine \textsc{Search\_Agent\_Red} (according to steps \ref{searchagentblue:9}-\ref{searchagentblue:10}). Using similar arguments as in Case \ref{case 3b} above, we can say that no collision happens between the two agents when $A_2$ moves from $y$ to $u$.

            \item If $A_1$ is not in the neighborhood and $A_2$ finds a blue pebble, then $A_2$ calls subroutine \textsc{Search\_Agent\_Blue} (according to steps \ref{searchagentblue:7}-\ref{searchagentblue:8}). In this case, $A_2$ repeats the process until either $A_1$ is in the neighborhood or a red pebble is found. If $A_1$ is in the neighborhood, then by case \ref{case 3c} above and if a red pebble is found, then by \ref{case 3d} above, we know that no collision happens between the two agents when $A_2$ moves from $y$ to $u$.
        \end{enumerate}
    \end{enumerate}
\end{proof}

\begin{lemma}\label{lem:4_4}
    If $deg(y) < n-1$ and $deg(x) < n-1$, then $A_2$ moves from $u$ to $z_x$ without collision only after $A_2$ knows that $A_1$ is awake. 
\end{lemma}

\begin{proof}
    According to lines \ref{initline1} and \ref{initline2} of Algorithm \ref{alg: init} and lines \ref{blue1} and \ref{blue2} of subroutine \textsc{Search\_Agent\_Blue}, $A_2$ reaches node $u$ by either executing subroutine \textsc{Agent\_Found} or executing subroutine \textsc{Search\_Agent\_Red}. Based on the scenario of how $A_2$ reaches $u$, we consider the following two cases.
    \begin{enumerate}
        \item $A_2$ reaches node $u$ while executing subroutine \textsc{Agent\_Found}. In this case, $x$ is a neighbor of some nodes on the path $P$. By definition of $u$, the next node $v$ on path $P$ is not a neighbor of $x$. The agent first finds $v$ then periodically visit $u$ from $v$ in order to let $A_1$ know that $A_2$ is already waked up. This is achieved by steps \ref{agentfound:1}-\ref{agentfound:2} of subroutine \textsc{Agent\_Found} where $A_2$ visits neighbors of $u$ in increasing port numbers if $u$ has a blue pebble and in decreasing port numbers if $u$ has a red pebble until it finds a node with a red or blue pebble. Hence, according to line \ref{agentfound:5}-\ref{agentfound:7}, $A_2$ calls subroutine \textsc{Notify\_Agent} from node $v$.

        Let $r$ be round when $A_2$ first calls subroutine \textsc{Notify\_Agent} from node $v$. According to steps \ref{notifyagent:1}-\ref{notifyagent:2} of subroutine \textsc{Notify\_Agent}, $A_2$ waits at $v$ for either $A_1$ to arrive at a node $u'$ (neighbor of $v$) or for $2n$ rounds whichever occurs first. If $A_1$ arrives at $u$ within $2n$ rounds, then $A_2$ knows that $A_1$ has awakened and ends subroutine \textsc{Notify\_Agent}. If $A_1$ does not arrive within $2n$ rounds, then according to steps \ref{notifyagent:7}-\ref{notifyagent:8}, $A_2$ moves to $u$ in the round $r+2n+1$ and moves back to $v$ in round $r+2n+2$ and repeats the process until $A_1$ arrives at $u'$. Hence, no collision occurs between the two agents.

        Once $A_2$ ends the execution of subroutine \textsc{Notify\_Agent}, it starts executing subroutine \textsc{Move\_to\_Z\_Blue}, if $v$ has a blue pebble so that $A_2$ can move to node $z_x$ or terminates if $v$ has a red pebble, showing that $A_2$ has reached node $z_x$. Hence, we can say that $A_2$ moves towards $z_x$ only after it knows $A_1$ has awakened.

        \item Agent $A_2$ finds node $u$ while executing subroutine\\ \textsc{Search\_Agent\_Red}. According to steps \ref{searchagentred:3}-\ref{searchagentred:4} of subroutine \textsc{Search\_Agent\_Red}, $A_2$ moves back to parent node $v$ of $u$. By definition of $u$, node $v$ is not a neighbor of $x$. Hence, according to line \ref{searchagentred:7}, $A_2$ calls subroutine \textsc{Notify\_Agent} from node $v$.

        Let $r$ be round when $A_2$ first calls subroutine \textsc{Notify\_Agent} from node $v$. According to steps \ref{notifyagent:1}-\ref{notifyagent:2} of subroutine \textsc{Notify\_Agent}, $A_2$ waits at $v$ for either $A_1$ to arrive at a node $u'$ (neighbor of $v$) or for $2n$ rounds whichever occurs first. If $A_1$ arrives at $u$ within $2n$ rounds, then $A_2$ knows that $A_1$ has awakened and ends subroutine \textsc{Notify\_Agent}. If $A_1$ does not arrive within $2n$ rounds, then according to steps \ref{notifyagent:7}-\ref{notifyagent:8}, $A_2$ moves to $u$ in the round $r+2n+1$ and moves back to $v$ in round $r+2n+2$ and repeats the process until $A_1$ arrives at $u'$. Hence, no collision occurs between the two agents.

        Once $A_2$ ends the execution of subroutine \textsc{Notify\_Agent}, it starts executing subroutine \textsc{Return\_Back} so that $A_2$ can return to node $y$ and then calls subroutine \textsc{Move\_to\_Z\_Red} to move to node $z_x$. Hence, we can say that $A_2$ moves towards $z_x$ only after it knows $A_1$ has awakened.
    \end{enumerate} 
\end{proof}

\begin{lemma}\label{lem:4_5}
    If $deg(y) < n-1$ and $deg(x)<n-1$, then by following Algorithm \ref{alg: init}, there is no collision between the agents when $A_2$ moves from $u$ to $z_x$.
\end{lemma}

\begin{proof}
    The agent $A_2$ starts moving from $u$ to $z_x$ only after learning that $A_1$ is awake. Now, $A_2$ moves to $z_x$ either directly from $u$ by executing subroutine \textsc{Move\_To\_Z\_Blue}, or it moves from $u$ to $y$ using subroutine \textsc{Return\_Back} and then from $y$ to $z_x$ using subroutine \textsc{Move\_To\_Z\_Red}. In both cases, in the path from $u$ to $z_x$, $x$ is not present. If $x$ is a neighbor of some node in the path $P$ or $P_1$, our pebble placement strategy and the corresponding movement algorithm guarantees that while moving along $P$ (or $P_1$), from a node on $P$ (or $P_1$), the next node in the path, where a pebble of color blue or red is placed will be visited before $x$. On the other hand, once $A_2$ starts its movement from $u$ to $z_x$, $A_1$ waits at $x$ for $7n$ rounds. Therefore, it is enough to show that $A_2$ reaches $z_x$ within $7n$ rounds once its starts its movement from $u$. In order to prove this, we show that at most $6n$ round is needed to reach from $y$ to $z_x$ and at most $n$ round is needed to reach from $u$ to $y$. 

    Let the number of nodes on path $P$ be $a$. Then the number of nodes not on $P = n-a$. For any node $w\notin P$, $w$ is a neighbor of at most 3 nodes of $P$. Therefore, for every node $w\notin P$, there are at most 3 edges connecting it to nodes in $P$. Hence, for $n-a$ nodes not in $P$, there are at most $3(n-a)$ edges connecting it to nodes on $P$. If the agent at node $v_i\in P$ visits its neighbor $w\notin P$, then the edge $(v_i,w)$ is traversed twice, once to visit $w$ and once to return to $v_i$. Hence, at most $3(n-a)$ edges not in $P$ are traversed at most $6(n-a)$ times. If agent at node $v_i\in P$ visits its neighbor $v_{i+1}\in V_1$, then the edge $(v_i,v_{i+1})$ is traversed only once to visit $v_{i+1}$. Hence, $a-1$ edges in $P$ are traversed $a-1$ times. Therefore, to reach from $y$ to $z_x$ total time needed is at most $6n-5a+1<6n$. On the other hand, since $u$ is reached from $y$ by the agent $A_2$, it already knows its path back to $y$ (by storing the incoming ports while traveling from $y$ to $u$). Hence, it can reach from $u$ to $y$ in time $dist(u,y)$ which is at most $n$. This completes the proof of the lemma.
\end{proof}

The next lemma shows that once $A_1$ starts its movement after waiting $7n$ rounds, it successfully reaches to $z'_x$ without any collision with $A_2$.
\begin{lemma}\label{lem:4_7}
    If $deg(y) < n-1$ and $deg(x)<n-1$, then the agent $A_1$ without any collision identifies the fact that $A_2$ is awake and then successfully reaches to $z'_x$ without any collision.
\end{lemma}
    
\begin{proof}
    According to Algorithm \ref{alg: init}, if no pebble is found at the initial node, the agent calls subroutine \textsc{PortLabeledBFS}. 

    There are two major steps executed by subroutine \textsc{PortLabeledBFS}. In the first step, $A_1$, after waking up, identifies whether $A_2$ is in the neighborhood and $A_2$ is active. In the second step, after waiting for $7n$ rounds, it explores the nodes in $G$ using BFS traversal until a node with a green pebble is found. 
 
    We start by showing no collision between the agents in the first step. Based on what the agent $A_1$ observes after waking up, we consider the following two cases.
 
    \begin{enumerate}
        \item If $A_2$ is in the neighborhood of $A_1$, then according to lines \ref{bfs:3}-\ref{bfs:8} of \textsc{PortLabeledBFS}, $A_1$ waits at $x$ for 2 rounds and then checks if $A_2$ is still in the neighborhood. If $A_2$ is in the neighborhood, then $A_1$ continues to wait at $x$ for 2 rounds until $A_2$ is not in the neighborhood. Once $A_2$ leaves the neighborhood, $A_1$ knows that $A_2$ has awakened and $A_2$ is 2 distance apart from $A_1$. It then visits the neighborhood of $A_2$ to inform about its awakened state. Hence it does not collide with $A_2$.
 
        \item If $A_2$ is not in the neighborhood of $A_1$, then according to lines \ref{bfs:4}-\ref{bfs:2} of \textsc{PortLabeledBFS}, $A_1$ waits at $x$ until $A_2$ arrives in the neighborhood to inform $A_1$ about its awakened state. $A_1$ then waits at $x$ for 2 rounds so that $A_2$ leaves the neighborhood of $x$. Hence, $A_2$ is 2 distance apart from $A_1$. It then visits the neighborhood of $A_2$ to inform about its awakened state. Hence it does not collide with the other agent.
    \end{enumerate}

    Once $A_1$ knows $A_2$ is awakened and informs $A_2$ about its awakened state, it waits for $7n$ rounds. As mentioned in the proof of lemma \ref{lem:4_5}, the agent $A_2$ reaches $z_x$ within round $7n$. Hence, it is enough to prove that $A_2$ never visits $z_x$ before it terminates.
 
    As per the pebble placement strategy, a green pebble is placed at $z'_x$ and $A_1$ continues visiting the nodes of $G$ according to BFS until it visits a node with a green pebble. As per the definition of $z'_x$ and $z_x$, the BFS traversal from $x$ visits $z'_x$ before $z_x$. Since $A_1$ terminates at $z'_x$, it does not visit $z_x$ and hence no collision occurs.
\end{proof}

Finally, we are ready to prove the final theorem of this section that gives the correctness of the proposed algorithm. 
\begin{theorem}
 By the time when both agents terminate, each node $v \in G$ is visited by at least one agent and no collision happens between the two agents.
\end{theorem}
    
\begin{proof}

 Based on the degrees of the nodes $x$ and $y$, we consider the following two cases.
 \begin{enumerate}
 \item Suppose $deg(y) = n-1$ or $deg(x) = n-1$, then according to lemma \ref{lem:4_3}, each node $v \in G$ is visited by at least one agent and no collision happens between the two agents.

 \item Suppose $deg(y) < n-1$ and $deg(x) < n-1$, then $A_2$ visits the nodes on the path $P$ and $P_1$. Let $V'$ be the set of nodes on either $P$ or $P_1$ except $x$. By definition of $z'_x$, $z'_x$ is the node that is visited last among all the nodes that are neither in $P$ nor in $P_1$. Hence, before $A_1$ finds a green pebble at $z'_x$, it visits all the nodes that are not on $P$ or $P_1$. Therefore, each node $v \in G$ is visited by at least one agent, and from the lemmas \ref{lem:4_5} and \ref{lem:4_7}, no collision happens between the two agents.
 \end{enumerate}
\end{proof}

\subsection{Polynomial exploration for Bipartite Graphs}

The algorithm proposed in section \ref{sec:4} takes $O(\Delta^D)$ time before both agents terminate and guarantees collision-free exploration with color index 3. In this section, we show that if the input graph is bipartite, then the previous algorithm can be modified to guarantee collision-free exploration in polynomial time with color index 2. In fact, in the case of a bipartite graph, the proposed exploration algorithm achieves a stronger version of the exploration problem where each node is required to be visited by both agents.

Before we describe the details of the algorithm, we give a high level idea of the same. The exploration part of the algorithm is based on the universal exploration sequence \cite{KOUCKY2002}. From the result in \cite{Reingold2008}, if the agent knows $n$, it is possible to compute a sequence of port numbers of length polynomial in $n$ such that if the agent follows this sequence of port numbers, then each node of any graph of size $n$ can be visited at least once. We call this sequence $ExploreUniversal(n)$. 

The modified algorithm runs in three phases. The first two phases are the same as before, where the objectives are to move one agent to the neighborhood of the other agent, and both agents learn once the other agent has woken up. Once these two phases are done, in the third phase, both agents start exploring the graph at the same time according to a universal exploration sequence and terminate once the exploration is completed. Since the graph is bipartite and the agents are at a distance 1 when they start the exploration in phase 3, and both agents always move along an edge before they terminate, they never collide and every node in the graph is explored. The details of the pebble placement and the exploration algorithm are described below.


{\bf Pebble placement:} Let $x$ and $y$ be the initial positions of the agents such that $deg(x)\ge deg(y)$. 
\begin{enumerate}
 \item If $deg(x) = n-1$ and $deg(y) = 1$, and $y$ is connected to $x$ through port 0, then place a red pebble at $x$. Otherwise, place a black pebble at $x$.
 \item If $deg(x) \le n-1$ and $deg(y) \le n-1$, then consider the following cases.
 \begin{enumerate}
 \item If $dist(x,y)>1$, and $u$ is an adjacent node of $y$ on the shortest path from $x$ to $y$, then for every node $v$ on the shortest path from $x$ to $u$, place a black pebble.
 \item If $dist(x,y)=1$, and $y$ is connected to $x$ through port 0, then a red pebble is placed at $x$. Else a black pebble is placed at $x$.
 \end{enumerate}
\end{enumerate}

{\bf Algorithm:} The algorithm for the agents is described as follows:

\begin{enumerate}
 \item If an agent after waking up sees no pebble, then
 \begin{enumerate}
 \item If the agent finds the other agent in its neighborhood, it waits until the other agent has moved out of its neighborhood. Execute step \ref{step1.1}.
 \item Otherwise, it waits until the other agent arrives in the neighborhood. Once the other agent arrives in the neighborhood, execute step \ref{step1.1}. 
 \item \label{step1.1} It visits its neighbors until it reaches the node $u$ such that the other agent is in the neighborhood of $u$. Then it returns to node $y$ and executes $ExploreUniversal(n)$.
 \end{enumerate}

 \item If the agent, after waking up, finds a pebble, then
 \begin{enumerate}
 \item \label{step2.1} If the other agent is not in the neighborhood, then execute step \ref{step2.2}. Else execute step \ref{step2.3}.
 \item \label{step2.2} It visits its neighbors one by one in increasing order of port numbers, starting from port 0 until it finds a node with a black pebble. Go to step \ref{step2.1}.
 \item \label{step2.3} Let the current node be $u$. If $u\ne x$, then move back to the previous node from which it arrived $u$ by executing step \ref{step2.2} and then execute the following steps. Otherwise, execute step \ref{step2.4}.
 \begin{enumerate}
 \item \label{st1.1} Stay in the current node for $2n$ rounds or until the other agent comes in the neighborhood, whichever occurs first.
 \item \label{st1.2} If the other agent arrives in the neighborhood within this period, move back to node $u$ and start executing $ExploreUniversal(n)$. Otherwise, move back to node $u$ and execute step \ref{step2.3}.
 \end{enumerate}
 \item \label{step2.4} If the agent finds a black pebble at $x$, then it moves to its neighbor which is connected to $x$ by port 0, else if it finds a red pebble at $x$, then it moves to its neighbor which is connected to $x$ by port 1. Then it executes steps \ref{st1.1} and \ref{st1.2}.
 \end{enumerate}
\end{enumerate}

\begin{theorem}
 The proposed algorithm guarantees collision-free exploration in polynomial time in bipartite graphs.
\end{theorem}
    
\begin{proof}

The proposed algorithm has three major steps. The objective of the first step is to bring two agents at a distance 1 if they are at least two distances apart from each other. The objective of the second step is to learn a common round by which both agents are awake and then in the third step, both agents start following a universal traversal sequence of the same length and terminate at the same time. As the length of the universal traversal sequence is polynomial \cite{}, and the agents start from adjacent nodes of a bipartite graph, each node in the graph will be visited by both agents at least once and no collision will happen in the third step. Hence, it is enough to prove that the first two steps take polynomial time and no collision happens in these steps. Below, we consider the different scenarios of based on what the agent sees once wakes up.

\begin{itemize}
\item Suppose the agent, after waking up, finds a pebble at its initial node and the other agent in its neighborhood. In this case, the initial distance between the two agents is 1 and according to the pebble placement, this node is the node $x$. According to the algorithm, if the agent finds a black pebble at $x$, then $y$ is connected to node $x$ through port $p\ne 0$. In that case, the agent moves to its neighbor node $u$ using port 0 and waits for $2n$ rounds. According to the algorithm, if the agent finds a red pebble at $x$, then $y$ is connected to node $x$ through port 0. In that case, the agent moves to its neighbor node $u$ using port 1 and waits for $2n$ rounds. If the other agent at $y$ is awake at that time, then it visits its neighbor until it reaches a node that is a neighbor of node $u$. Since the graph is bipartite, any node that is a neighbor of node $x$ cannot be the neighbor of node $y$ if $dist(x,y) = 1$. Hence, no collision happens when the agent moves from $y$ to its adjacent nodes.

\item Suppose the agent, after waking up, finds a pebble at its initial node and the other agent is not in its neighborhood. In that case, the initial distance between the two agents is $>1$ and according to the pebble placement, this node is node $x$. According to the algorithm, the agent visits its neighbor in the increasing port number until it finds a node $v$ with a black pebble. If the other agent is not in the neighborhood of $v$, then the agent visits the neighbor of $v$ in the increasing port number until it finds a node with a black pebble. The agent continues this process until it finds a node $u$ with a black pebble and the other agent in the neighborhood of $u$. In this case, the agent moves to its previous node and waits for $2n$ rounds. If the other agent at $y$ is awake at that time, then it visits its neighbor until it reaches a node that is a neighbor of node $u$. Since the graph is bipartite, any node that is a neighbor of node $x$ cannot be the neighbor of node $y$ if $dist(x,y) = 1$. Hence, no collision happens when the agent moves from $y$ to its adjacent nodes.
\end{itemize}

In case the agents are at least two distances apart, the time taken by the agent at $x$ to reach a neighbor of $y$ takes $O(m)$ times where $m$ is the number of edges in the graph. Once the agents are in the neighborhood, within $2n$ rounds in case both agents are awake this time, both agents learn the other agent is awake. In case the agent of $y$ is still not awake, within $4n$ rounds after the agent at $y$ wakes up, both agents learn the other agent is awake. Hence, both step 1 and step 2 are executed in polynomial time without any collision. This completes the proof.

\end{proof}

\section{Conclusion}
In this paper, we have studied the collision-free exploration problem by two mobile agents with the help of efficient pebble placement. We show an impossibility result showing there is no algorithm that solves the problem for color index 1. We have proposed two algorithms: an algorithm for a general graph with color index 3, and an algorithm for a bipartite graph for color index 2. The second algorithm runs in polynomial time. Whether there exists an algorithm for pebble placement with color index two remains open. On the other hand, for color index 3, as the proposed algorithm runs in exponential time, it remains to investigate what is the fastest algorithm that guarantees collision-free exploration for the general graph.

\bibliography{ref}
\end{document}